\documentclass[conference]{ieeeconf}
\usepackage{graphicx}
\usepackage{tcolorbox}
\AtBeginDocument{%
  \providecommand\BibTeX{{%
    \normalfont B\kern-0.5em{\scshape i\kern-0.25em b}\kern-0.8em\TeX}}}

\usepackage{graphicx}
\usepackage{wrapfig}
\usepackage{booktabs} % For formal tables

\usepackage{url}
\usepackage{algorithm}
\usepackage[noend]{algpseudocode}
\usepackage{amssymb,amsmath}
\usepackage{subcaption}
\usepackage{graphicx}
\usepackage{comment}

\usepackage{amsthm}
\usepackage{thmtools, thm-restate}
\usepackage{multirow}

\usepackage[graphicx]{realboxes}

\usepackage{xcolor} % Required for specifying colors by name
\definecolor{ocre}{RGB}{243,102,25} % Define the orange color used for highlighting throughout the book

\usepackage{avant} % Use the Avantgarde font for headings
\usepackage{mathptmx} % Use the Adobe Times Roman as the default text font together with math symbols from the Sym­bol, Chancery and Com­puter Modern fonts

\usepackage{microtype} % Slightly tweak font spacing for aesthetics
\usepackage[utf8x]{inputenc} % Required for including letters with accents
\usepackage{tikz}
\usetikzlibrary{shapes}
\RequirePackage{ifthen}
\RequirePackage{amsmath}
\RequirePackage{amsthm} %looks much better with this
\RequirePackage{amssymb}
\RequirePackage{xspace}
\RequirePackage{graphics}
\usepackage{color}
\usepackage{keyval}
\usepackage{xspace}
\usepackage{mathrsfs}

\DeclareMathAlphabet{\mathcal}{OMS}{cmsy}{m}{n}

\newcommand{\hussein}[1]{\textcolor{black}{#1}}

\newcommand{\reals}{{\mathbb{R}}}

\newcommand{\naturals}{{\mathbb{N}}} 
\newcommand{\dom}{\relax\ifmmode {\mathit{dom}} \else ${\sf dom}$\fi}

\newcommand\blfootnote[1]{
    \begingroup
    \renewcommand\thefootnote{}\footnote{#1}
    \addtocounter{footnote}{-1}
    \endgroup
}

\usepackage{setspace}
\usepackage{listings}
\newcounter{theorems}
\newtheorem{theorem}[theorems]{Theorem}%[theorem]

 \newcounter{remarks}
 \newcounter{assumptions}
\newtheorem{remark}[remarks]{Remark}%[remark]
\newtheorem{assumption}[assumptions]{Assumption} %[assumption]
 \newcounter{definitions}
 \newtheorem{definition}[definitions]{Definition}%[definition]

\newtheorem*{notation}{Notation}
\usepackage[normalem]{ulem}

\usepackage{placeins}
\makeatletter
\newcommand\nocaption{%
    \renewcommand\p@subfigure{}
    \renewcommand\thesubfigure{\thefigure\alph{subfigure}}
}
\makeatother

\title{Safe Decentralized Multi-Agent Control using Black-Box Predictors, Conformal Decision Policies, and Control Barrier Functions}
\author{
	Sacha Huriot and Hussein Sibai \\
	Computer Science and Engineering Department\\
	Washington University in St. Louis \\
	\texttt{\{h.sacha,sibai\}@wustl.edu}
}
\begin{document}

\maketitle

\begin{abstract}

We address the challenge of safe control in decentralized multi-agent robotic settings, where agents use uncertain black-box models to predict other agents' trajectories. We use the recently proposed conformal decision theory to adapt the restrictiveness of control barrier functions-based safety constraints based on observed prediction errors. We use these constraints to synthesize  controllers that balance between the objectives of safety and task accomplishment, despite the prediction errors. 
We provide an upper bound on the average over time of the value of a monotonic function of the difference between the safety constraint based on the predicted trajectories and the constraint based on the ground truth ones. We validate our theory through experimental results showing the performance of our  controllers when navigating a robot in the multi-agent scenes in the Stanford Drone Dataset.\blfootnote{This project was supported by the NSF CPS award No. 2403758.}
\end{abstract}
\section{Introduction}
Control in decentralized multi-agent settings is a fundamental problem with abundant applications in various domains, e.g., autonomous vehicles~\cite{10.1145/3582576}, power grids~\cite{SAMADI2020106211}, and manufacturing~\cite{9926465}. Because of the absence of communication between agents in such settings, an ego agent has to rely on sensing other agents' states and predicting their behaviors to plan its own control and achieve its objectives.
The problems of obtaining accurate state estimates and predictions have received enormous research attention that produced a myriad of successful methods and tools, particularly those adopting recent deep learning approaches~\cite{8643440, GAN, li2020gripplus,pedestrian_trajectory_prediction_using_CNNs_2022}. However, despite progress, prediction and perception are often susceptible to errors, and are expected to remain so. Thus, it is of utmost importance to design controllers that are aware of such errors and adapt their decisions accordingly.

Safety is an essential specification of almost any multi-agent system. It represents the requirement of avoiding states that the user deems unsafe, e.g., collisions. A prominent approach to guaranteeing safety of dynamical systems is using control barrier functions (CBFs) to filter out safety-violating nominal controls~\cite{CBF_based_quadratic_programs_2017_AaronAmes}. The super-level set of a CBF of a control system is {\em forward invariant}, i.e., all trajectories starting from this set remain inside it, if the controller always satisfies the CBF constraint. When the super-level set is disjoint from the set of states deemed unsafe, safety is implied. However, when a CBF is defined over the combined state of two agents, e.g., to specify their collision avoidance, the instantaneous dynamics of both agents have to be known for either agent to filter out its unsafe controls. In the absence of communication, an ego agent would have to resort to predicting the other agent's dynamics to estimate the corresponding CBF constraint. This might lead to unsafe controls being perceived as safe, and then, to safety violations.

Recently, conformal decision theory (CDT) has been proposed as a general framework for control using the uncertain outputs of black-box predictors~\cite{lekeufack2024conformal}. 
It is inspired by the theory of conformal prediction. The latter has proven to be an effective tool for distribution-free uncertainty quantification for black-box models~\cite{conformal_prediction_original_gammerman_1998,conformal_prediction_intro_anastasios_2023,conformal_decision_making_vovk_2018,adaptive_conformal_inference_under_distribution_shift_gibbs_neurips_2021}. It has gained increased interest in recent years with the rise of deep learning models deployed in critical settings~\cite{dixit2022adaptive,safety_assurance_using_conformal_prediction_2023}. CDT relaxes all statistical assumptions made in the conformal prediction literature, which do not necessarily hold in decision-making scenarios where data are inherently dependent. The guarantees provided by CDT hold even in adversarial settings, where predictions are intentionally made to worsen the performance of the controller.
CDT introduces a variable, called a {\em conformal variable}, that parameterizes control policies. It assumes that tuning that variable alters the conservatism of the followed policy. It thus tunes it  based on a user-defined {\em loss function}, quantifying how prediction errors affect performance and safety, rather than building and tuning prediction sets as we do when using conformal prediction. It assumes that there exists a conservative enough policy that would decrease the average value of the loss below a predefined threshold if followed at least for a predefined horizon, and it provides an upper bound on that average loss.

In this paper, we follow the CDT approach in addressing the challenge of control in decentralized multi-agent setting with collision avoidance safety specifications. We assume that an ego agent is equipped with a black-box predictor that estimates surrounding agents' trajectories periodically. At the end of a period, its sensors capture  the actual trajectories of these agents. We define CBF-based collision avoidance constraints that are strengthened or loosened based on observed prediction errors. We tune the constraints by updating the value of the conformal variable, that is introduced in an additive manner in the CBF constraint. We obtain an upper bound on the average of a monotonic function of the difference between the CBF constraints based on predicted and ground truth states, respectively. We also present experimental results validating our approach in maintaining the safety of a robot navigating different scenes in the Stanford Drone Dataset (SDD)~\cite{SDD_2016} while using a predictor to estimate the trajectories of  surrounding pedestrians.    
\subsection{Related Work}
\label{sec:RelatedWork}
Online adaptation of CBFs has been explored to improve the probability of safety in uncertain environments \cite{Risk-Aware-CBF}, to navigate and interact with environments that fixed CBF cannot do safely \cite{onlineCBFdecenter}, and for efficient task accomplishment while maintaining safety in multi-agent settings by estimating trust~\cite{rate_tune_CBF}. Conformal prediction has been used for online adaptation of a learned control policy to out-of-distribution states~\cite{huang2023conformal}, and for safe motion planning in multi-agent settings to define uncertainty regions around predictions~\cite{dixit2022adaptive}.

\section{Preliminaries}
\label{sec:Preliminaries}

\begin{notation}
    We denote the set of positive integers by $\naturals^{>0}$ and the set of non-negative reals by $\reals^{\geq0}$. We will use $t\in\reals^{\geq0}$ and $k\in\naturals^{>0}$ for continuous and discrete time instances, respectively. A function $\alpha\colon (-b,a)\rightarrow\reals$, for some $a$ and $b > 0$, belongs to the extended class-$\mathcal{K}$ if it is continuous, strictly increasing, and $\alpha(0)=0$~\cite{khalil2002nonlinear}. For any function $h:X\to\reals$, and any $r \in \reals$, we denote the $r$-super level set of $h$ by $h_{\geq r} := \{x\in X\mid h(x) \geq r\}$. Additionally, we denote the boundary of $h_{\geq r}$ by $h_{=r}$ and its interior by $h_{>r}$.
\end{notation}
Throughout the paper, we consider systems with nonlinear affine control dynamics of the form:
\begin{align}
    \dot{x} =f(x)+g(x) u,\label{eq:dynamics}
\end{align}
where $f:\reals^n\rightarrow\reals^n$ and $g:\reals^n\rightarrow\reals^{n\times m}$ are locally Lipschitz continuous and $u:\reals^{\geq0}\rightarrow U \subset \reals^m$ is a piecewise-continuous function. We assume that the dynamics are forward complete, ensuring the existence of solutions globally in time.

\subsection{Quadratic programming for CBF-based safe control}

\begin{definition}(Zeroing control barrier function)
    A differentiable function $h:\reals^n \rightarrow\reals$ is called a zeroing control barrier function (CBF) for system~(\ref{eq:dynamics}) if there exists a locally Lipschitz extended class-$\mathcal{K}$ function $\alpha$ such that for some
    super-level set $\mathcal{D} := h_{\geq 0}$, with $c >0$, 
    for any $x\in \mathcal{D}$,
    \begin{align}
        \sup_{u\in U} [\mathcal{L}_f h(x) + \mathcal{L}_g h(x) u] \geq -\alpha(h(x)),\label{eq:cbf_def}
    \end{align}
    where $\mathcal{L}_f h(x) := \frac{\partial h}{\partial x}f(x)$ and $\mathcal{L}_g h(x) := \frac{\partial h}{\partial x}g(x)$ are the Lie derivatives of $h$ w.r.t. the dynamics.
    
\end{definition}
An implication of condition~(\ref{eq:cbf_def}) is that $\forall x \in h_{=0}, \sup_{u\in U} \dot{h}(x)\geq0$. Consequently,
the super-level set $h_{\geq 0}$ of $h$ is forward invariant by Nagumo’s theorem~\cite{nagumo_thrm}. Additionally, $h_{\geq 0}$ is locally asymptotically stable over $\mathcal{D}$ and the function  $V:\mathcal{D} \rightarrow \reals^{\geq 0}$ that is equal to $-h(x)$ over $\mathcal{D} \backslash h_{\geq 0}$ and equal to zero in $h_{\geq 0}$ is a corresponding Lyapunov function~\cite{robust_CBF_IFAC_2015}. Thus, all sub-level sets of $V$ in $\mathcal{D}$, or equivalently, all super-level sets of $h$ in $\mathcal{D}$ with negative levels, are forward invariant. This is not necessarily the case for super-level sets of $h$ with positive levels.
A CBF $h$ for system~(\ref{eq:dynamics}) serves as a safety filter, altering nominal controls to guarantee safety. At any state $x \in \mathcal{D}$, a common objective is to find the closest control $u \in U$ that satisfies condition~(\ref{eq:cbf_def}) to some reference control $u_{ref}$ in terms of euclidean distance.  
Since the objective is quadratic and the constraint represented by (\ref{eq:cbf_def}) is affine in $u$, the resulting optimization problem can be formulated as a {\em Quadratic Programming} (QP) one:
\begin{align}
    u_{QP}(x)\ &:=\ \underset{u\in U}{\mathrm{argmin}}\ \Vert u - u_{ref}(x,t) \Vert^2\label{eq:CBF_QP}\\
    &\text{s.t. }\ \mathcal{L}_f h(x) + \mathcal{L}_g h(x)  u+\alpha(h(x))\geq0.\notag
\end{align}

We say that an instance of a QP problem is {\em feasible} if the set of controls that satisfy the constraint in (\ref{eq:CBF_QP}) is not empty. 
When $U =\reals^m$ and $\forall x\in\mathcal{D},\ \mathcal{L}_gh(x)\neq0$, (\ref{eq:CBF_QP}) has a locally Lipschitz continuous closed-form solution $u_{QP}:\mathcal{D}\rightarrow\reals^m$~\cite{Comparison_CBF_APF_2020_AaronAmes}.

\subsection{Conformal decision theory}
In CDT, the decision-making agent is able to observe the ground truth of the predictions in a delayed fashion. At step $k+1$, it is able to observe the ground truth at step $k$. That would allow it to update its decision-making based on the difference between the observed ground truth and the prediction. Thus, designing controllers based on this theory requires a loss function that quantifies the quality of the decision after observing the ground truth, an update rule of the conformal variable based on the loss, and a family of controllers parameterized by the conformal variable. We describe them more formally next.
\begin{definition}The components of controllers based on CDT~\cite{lekeufack2024conformal}, called {\em conformal controllers},  are:
    \begin{itemize}
    \item an input space $\mathcal{X}$ to the controller (e.g., ego agent's state and the predictions of the instantaneous dynamics of surrounding agents),
    \item an action space $\mathcal{U}$ of the controller (e.g., the actuation space $U$ of system~(\ref{eq:dynamics})),
    \item a ground truth space $\mathcal{Y}$ (usually equal to $\mathcal{X}$), 
     \item a conformal variable $\lambda \in \reals$, which is updated  at every discrete time step according to the loss function values at previous steps,
     \item a loss function $L: \mathcal{U} \times \mathcal{Y} \rightarrow[0,1]$ that quantifies the quality of a decision in $\mathcal{U}$ based on a prediction after observing its ground truth in $\mathcal{Y}$, and 
     
    \item a family $\{D_k\}_{k\in \naturals}$ of feedback controllers, or {\em decision functions},  available for the agent at each time step of the form $D_k := \{ D_k^\lambda : \mathcal{X}\rightarrow \mathcal{U}\ |\ \lambda \in \reals \}$, where one should think of $D_k^{\lambda_1}$ as more conservative than $D_k^{\lambda_2}$, i.e., more likely to result in a smaller loss, when $\lambda_1 < \lambda_2$. 
        
    \end{itemize}
\end{definition}
Tuning $\lambda$ based on the loss plays the same role as that of  tuning the size of the prediction sets in conformal prediction for time series~\cite{adaptive_conformal_inference_under_distribution_shift_gibbs_neurips_2021,conf_pred_time_series}. To be able to provide guarantees on the performance of a conformal controller under prediction errors, one has to assume the existence of decision functions that are conservative enough to drive the loss down. In \cite{lekeufack2024conformal}, the authors assume the existence of such functions that can drive the average loss over some pre-defined time horizon below a user-defined threshold over all possible sequences of pairs of predictions and ground truths. This is formulated in the following definition. 

\begin{definition}(Eventually safe conformal controller) \label{def:eventually_safe}
    A conformal controller is eventually safe if $\exists \epsilon_\mathit{safe}\in[0,1],\ \lambda_\mathit{safe}\in\reals$, and a time horizon $K \in \naturals^{>0}$, such that we have, uniformly over all sequences $\lambda_1,\dots,\lambda_K$ and $(x_1,y_1),\dots,(x_K,y_K)$:
    \begin{align*}
        \left(\forall k\in[K],\ \lambda_k\leq\lambda_\mathit{safe}\right)\implies
        \frac{1}{K}\sum_{k=1}^{K}L(D_k^{\lambda_k}(x_k),y_k)&\leq\epsilon_\mathit{safe}.
    \end{align*}
\end{definition}
When such an assumption is satisfied, an update rule for the conformal variable can be derived to achieve a {\em long-term bound on the risk}, i.e., an upper bound on the average loss over large enough time horizons~\cite{lekeufack2024conformal}. This is formalized in the following theorem.
\begin{theorem}(Long-term risk bound~\cite{lekeufack2024conformal}) \label{thm:conformal_bound}
    Fix a user-defined $\epsilon\in [0,1]$, a {\em learning rate} $\eta \in \reals^{>0}$, an eventually safe conformal controller, and consider the following update rule for the conformal control variable:
    $\forall k\in\naturals^{>0},\ \lambda_{k+1}=\lambda_k +\eta(\epsilon-l_k)$,
    where $l_k=L(D_k^{\lambda_k}(x_k),y_k)$.
    If $\lambda_1\geq\lambda_\mathit{safe}-\eta$ and $\epsilon_\mathit{safe}\leq\epsilon$, then for any realization of the data, i.e., for any possible sequence of pairs of predictions and ground truths $(x_1,y_1),(x_2,y_2)\dots$, the average loss is bounded as follows:
    $$\forall K'\geq K,\ \frac{1}{K'}\sum_{k=1}^{K'}l_k\leq\epsilon+\frac{\lambda_1-\lambda_\mathit{safe}+K\eta}{\eta\cdot K'}.$$ Hence, it results in an $\epsilon+o(1)$ average loss in the long term, where $o(1)$ converges to zero as the horizon $K'$ increases.
\end{theorem}
\section{Problem setup}
\label{sec:ProblemSetup}
We consider a decentralized multi-agent setting, where the ego agent $i$, with known dynamics (\ref{eq:dynamics}), attempts to complete an objective while respecting safety.
The objective of agent $i$ is encoded by a reference controller $u_{ref}$ that might not be necessarily safe and end up leading to collisions. For example, one might consider a reference controller that satisfies the constraints of a control Lyapunov function (CLF), as in \cite{CBF_theory}, to stabilize to a certain equilibrium.
We will augment $u_{ref}$ with a safety filter based on CBFs.

\subsection{Pairwise collision avoidance}

We assume without loss of generality that all agents have the same state space $\reals^n$ and that a zeroing CBF $h:\reals^n \times \reals^n \rightarrow\reals$ that defines pairwise collision avoidance is user-provided. For each agent $j$, we define $h_j:=h(x_i,x_j)$.
We define its associated barrier condition as a boolean function, or a constraint, $C_j:\reals^n \times \reals^n \times\reals^m\times\reals^n\rightarrow\{\top, \bot\}$ that represents the satisfaction of the inequality $\dot{h}(x_i,x_j) \geq -\alpha(h(x_i, x_j))$.
We decompose $\dot{h}_j$ using the chain rule as follows:
\begin{align}
    C_j(x_i,x_j,u_i,\dot{x}_j)\ 
    \equiv\ \frac{\partial h_j}{\partial x_i}\dot{x}_i + \frac{\partial h_j}{\partial x_j}\dot{x}_j + \alpha(h_j) \geq 0.\label{eq:cbf_decomp}
\end{align}
\begin{remark}
        We use $u_i$ instead of $\dot{x}_i$ in the definition of $C_j$ since we know agent $i$'s dynamics, in contrast with that of agent $j$, and we need it to explicitly appear as we will  use $C_j$ later to solve for a $u_i\in\reals^m$ that satisfies it.
\end{remark}
To be able to bound the effect of prediction errors on safety later in the paper, we will need the following assumption.  
   \begin{assumption}\label{asm:derivative_bound}
The norm of the partial derivative $\frac{\partial h_j}{\partial x_j}$
  is globally bounded by $M_h$, and $\alpha$ is globally Lipschitz continuous with Lipschitz constant $M_\alpha$.
       \end{assumption}
\subsection{Sensors and black-box trajectory predictors}
In our decentralized setting, agent $i$ does not have instantaneous knowledge of the other agents' states and dynamics,
on which its safety constraints depend. It only has such knowledge of its own state and dynamics.
Instead, agent $i$'s sensor can, every $\tau$ seconds, accurately sample the trajectories of the agents within its range, that we call the {\em sensed agents}. 
This setting implicitly assumes that, if an agent $j$ is not observed by the sensor, then it is far enough so that the safety constraint $C_j$ can be ignored.
Precisely, at any time instant $(k+1)\tau$, the sensor returns the trajectories $\xi_j:I_k\to\reals^n$ of the sensed agents\footnote{We assume without loss of generality that all surrounding agents have been in the range in the last $\tau$ seconds.} over the time interval $I_k=[k\tau, (k+1)\tau)$.
% This output cannot be used directly to solve CBF constraints over this time interval, since it is delayed. 
%However,
The ego agent is equipped with a {\em black-box predictor} that estimates the agents' trajectories for the next $\tau$ seconds, i.e., returning $\hat\xi_j:I_{k}\to\reals^n$ at time instant $k\tau$. Trajectory prediction has been an active research area in robotics and control~\cite{deeplearning_predictor_multiagent,multi_rl_network,Trajectory_Multiagent_Distributed,Amirloo2022LatentFormerMT,Linear_predictors}. The ego agent then estimates $\dot\xi_j(t)$, at any sensed agent $j$ and $t \in [k\tau,(k+1)\tau)$, from its  predicted trajectory. 
%It can also use a dynamics predictor, instead of the trajectory predictor since instantaneous dynamics are used in the safety constraint $C_j$}.
At the end of each interval, it uses the sensor's output to compare the past predictions against the ground truth to measure the predictor's error.
Again, to be able to guarantee the existence of a controller that is conservative enough to achieve safety despite prediction errors, we will assume the existence of, possibly unknown, upper bounds on the predictor's errors in terms of both values and time derivatives. 
%\sacha{
This assumption holds in many common applications, e.g., when the agents have a bounded speed in a closed environment.
%}
\begin{assumption}(Bounded prediction error) \label{asm:predictor}
    There exists a {\em value bound} $E_v$ and a {\em dynamics bound} $E_d$, such that, for all $k \in \naturals$, and for each sensed agent $j$, the predicted trajectory $\hat{\xi}_j$ satisfies 
        $\sup_{t\in I_k}\left\Vert\hat{\xi}_j-\xi_j\right\Vert\leq E_v$ and $\sup_{t\in I_k}\left\Vert\hat{\dot\xi}_j-\dot\xi_j\right\Vert\leq E_d$.  
\end{assumption}

\noindent {\bf Problem statement:}
        Given a zeroing CBF $h$ that encodes pairwise collision avoidance, a sensor that observes neighboring agents' trajectories in periods of  $\tau$ seconds,
        an unreliable predictor of their trajectories, a reference controller that does not necessarily guarantee safety, design a controller that  adapts its safety constraints  based on prediction errors while following its reference controller.

\section{Algorithm}
\label{sec:Algorithm}
In this section, we present our proposed solution for the problem above. It is an algorithm that adapts the ego agent's CBF constraints using CDT \cite{lekeufack2024conformal}.
It does so by using the conformal variable as a slack variable in the pairwise collision avoidance CBF constraints. 

\subsection{CBF-based conformal controller}

    We build a conformal controller based on control barrier functions, as shown in Figure~\ref{fig:algo_diagram}. Its components are as follows:  
\begin{itemize}
    \item the input space $\mathcal{X}$ to our controller is the set of all triples, where the first element belongs to the output space $\mathcal{X}_p$ of the trajectory predictor, i.e., the set of all finite sets of possible predicted trajectories over an interval of $\tau$ seconds of a finite number of sensed agents, the second element is an instantaneous ego agent's state, and the third is a time instant $t$, 
    \item the action space $\mathcal{U}$ of the controller is the 
    % set of all piece-wise constant signals mapping the interval $[0,\tau)$ to the 
    control space of the ego agent $U$,
    \item the ground truth space $\mathcal{Y}$ is the set of all pairs of an actual trajectory of the ego agent of duration $\tau$ and a  set of actual trajectories of its sensed agents over the same interval,
    % in the environment, 
         \item the loss function $L: \mathcal{X}_p \times \mathcal{Y} \times \reals \rightarrow(-1/2,1/2)$ is defined as follows:
          $L( \{\hat{\xi}_j\}_{j \in [N]}, (\xi_i, \{\xi_j\}_{j \in [N]}), \lambda)=\max_{j \in [N],  t \in [0,\tau)} s\left(\hat{\dot{h}}_j+\alpha(\hat{h}_j) +\lambda -\dot{h}_j - \alpha(h_j)\right)$,
     where we use $\hat{h}_j := h(x_i,\hat{x}_j)$, and $\hat{\dot{h}}_j:= \frac{\partial \hat{h}_j}{\partial x_i}\dot{x}_i+\frac{\partial \hat{h}_j}{\partial x_j}\hat{\dot{x}}_j$, with $x_j$ and $\hat{x}_j$ denoting $\xi_j(t)$ and $\hat{\xi}_j(t)$, $\lambda$ is the conformal variable, and $s:r\in\reals\rightarrow(-1/2,1/2)$ is an extended class-$\mathcal{K}$ function, in order to preserve the sign and variations of its argument, such as the function $r\mapsto\arctan(r)/\pi$,  
          \item the conformal variable $\lambda \in \reals$ is updated  in the $(k+1)^{\mathit{th}}$ sampling time, i.e., at $t = (k+1)\tau$, according to the equation $\lambda_{k+1}=\lambda_k +\eta(\epsilon-l_k)$, where $\eta$ is a constant learning rate, $\epsilon$ is the user-defined target average loss, and $l_k$ is the loss over $I_k:=[k\tau,(k+1)\tau)$, as in \cite{lekeufack2024conformal},  and 
         \item the set $\{D_k\}_{k\in \naturals}$ of families of feedback controllers is defined by Algorithm~\ref{code:conformal_controller}. 
                                    \end{itemize}

Our conformal controller update its conformal variable $\lambda$ every $\tau$ seconds, at the same time instants the sensor and the predictor update their respective outputs, as shown in Figure~\ref{fig:algo_diagram}.
In contrast, its  decision functions are determined by a CBF-based QP solver that generates controls instantaneously in time, as described in Algorithm~\ref{code:conformal_controller}\footnote{In practical settings, this means that we run the QP at a much higher frequency than the sensor and the predictor.}. 
The latter    
uses the predicted trajectories of neighboring agents in an arbitrary interval $[k\tau,(k+1)\tau)$ to estimate their instantaneous dynamics at any time instant within it. It then uses these estimates to check the CBF constraints that represent pairwise collision avoidance with surrounding agents. 
A simple replacement of $\dot{h}_j$ by $\hat{\dot{h}}_j$ and $\alpha(h_j)$ by $\alpha(\hat{h}_j)$ in $C_j$ might lead to violating $C_j$ because of prediction errors.
In particular, such a violation will happen if and only if $\hat{\dot{h}}_j + \alpha(\hat{h}_j) \geq 0>\dot{h}_j + \alpha(h_j)$.
Instead, in addition to that replacement, 
our conformal controller has 
$\lambda$, the conformal variable, added to the left hand side of the inequality, i.e., it uses $\hat{C}_j:(\reals^n)^2\times \reals^m\times \reals^n\times\reals\rightarrow\{\top, \bot\}$ as a replacement of $C_j$, which we call the {\em conformal constraint}, that we define as follows: $\hat{C}_j(x_i,\hat{x}_j,u_i,\hat{\dot{x}}_j,\lambda) \equiv$
\begin{align}
\label{eq:conformal_constraint}
% \hat{C}_j(x_i,\hat{x}_j,u_i,\hat{\dot{x}}_j,\lambda)\quad\equiv\quad
\frac{\partial \hat{h}_j}{\partial x_i}\dot{x}_i + \frac{\partial \hat{h}_j}{\partial x_j}\hat{\dot{x}}_j+ \alpha(\hat{h}_j) + \lambda\geq0.
\end{align}
Hence, the value $\hat{\dot{h}}_j+ \alpha(\hat{h}_j) + \lambda- \dot{h}_j - \alpha(h_j)$ can serve as a measure of the effect of replacing the real  constraint $C_j$ with  $\hat{C}_j$ on safety.
It can be seen as the {\em gap} between the two constraints: if it's positive, it bounds the extent of possible constraint violation, and if it's negative, it lower bounds the extent of robust satisfaction of the real constraint. 
% represents a robust satisfaction of the real constraint.
% tolerance to prediction errors 
% without affecting safety. 
Updating $\lambda$ proportionally to such value allows adaptation of the constraints to prediction errors even when the chosen control satisfies $C_j$. This is what motivates us to choose the loss we defined earlier for our conformal controller. 

\begin{figure}[!h]
    \centering
    \includegraphics[width=\linewidth]{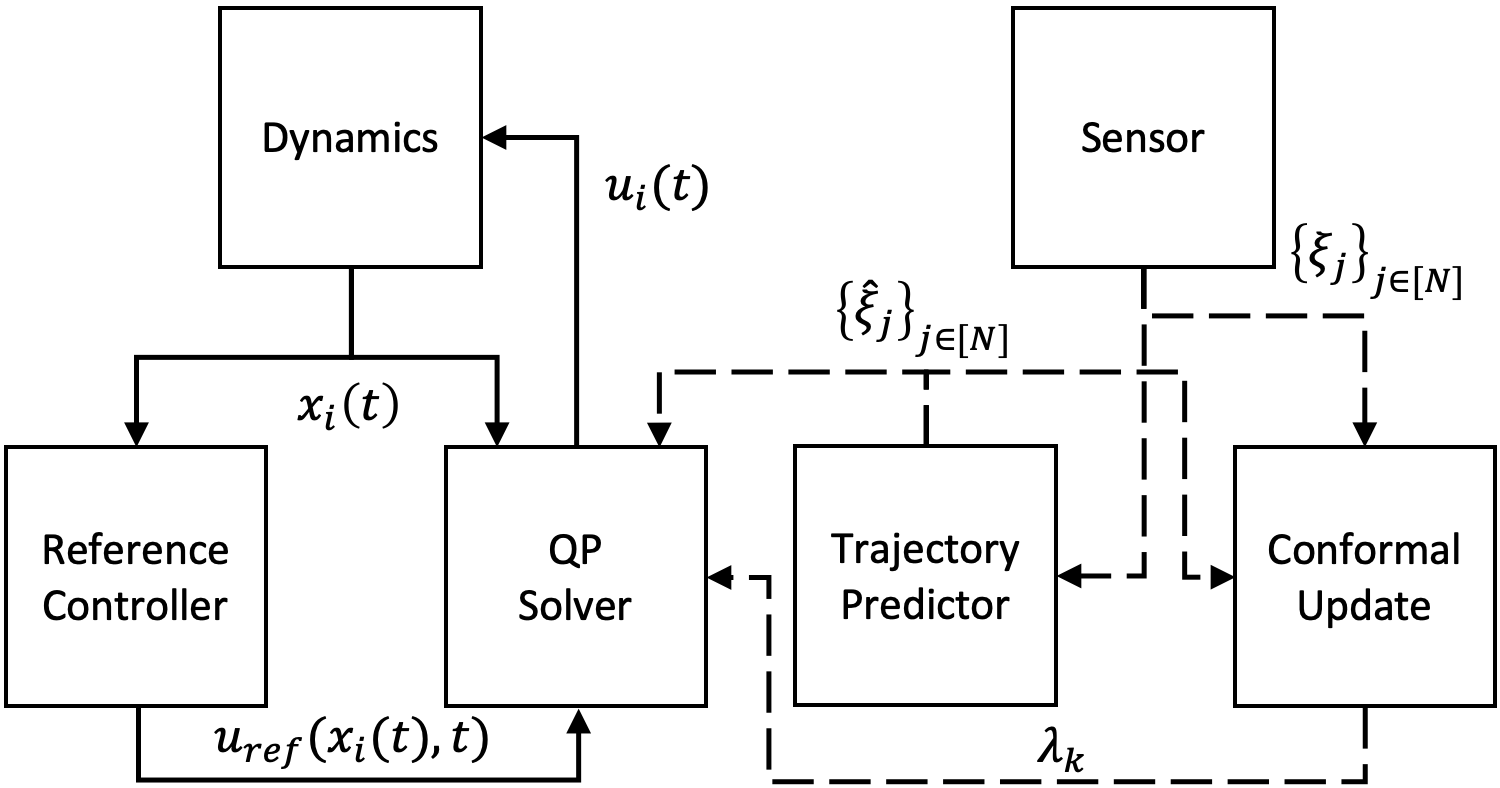}
    \caption{The reference controller and QP solver run in a feedback loop instantaneously (plain arrows), while the trajectory predictor and conformal update run periodically (dashed arrows). For $t\in I_k=[k\tau,(k+1)\tau)$, the control $u_i(t)$ is the minimal deviation $u$ from $u_{ref}(x_i(t),t)$ that satisfies the conformal safety constraints $\left\{\hat{C}_j(x_i,\xi_j(t),u,\hat{\dot{\xi}}_j(t),\lambda_k)\right\}_j$, where the predicted trajectories $\{\hat{\xi}_j\}_{j\in[N_{k}]}$ and the conformal parameter $\lambda_k$ were obtained at $k\tau$. To update the latter, we use the maximal prediction error between the ground truth trajectories $\{\xi_j\}_{j\in[N_{k-1}]}$ obtained from the sensor over $I_{k-1}$ and the predicted trajectories  at $(k-1)\tau$.
    }
    \label{fig:algo_diagram}
\end{figure}

A loss greater than the user-defined target average loss $\epsilon$ means that the constraint based on the current $\lambda$ and predicted trajectories is less conservative than needed. 
% predictor is less accurate than expected. 
Hence, our conformal controller decreases $\lambda$ to make $\hat{C}_j$ stricter, and vice versa.
Since our loss can be negative, we also allow $\epsilon$ to be negative. 
If the user chooses $\epsilon>0$, our update of $\lambda$ based on our loss drives   $\hat{C}_j$ to be looser than $C_j$, and thus allows some violations of the latter, and agent $i$ can better follow its reference control.
% \sacha{
If they chose an $\epsilon<0$, our update of $\lambda$ drives $\hat{C}_j$ to be stricter than $C_j$, and thus allows % Algorithm~\ref{code:conformal_controller} 
the satisfaction of $C_j$ more often on average.
%}
Since over any time interval $I_k$, there are multiple constraints, one per nearby agent, and that Algorithm~\ref{code:conformal_controller} is ran based on the predictions made at $k\tau$ over the whole interval, our conformal controller updates $\lambda$ at $(k+1)\tau$ based on the worst case prediction error over both time and agents. That is the reason we have the maximum operator in our definition of the loss function. % in our conformal controller.

\subsection{Conformal CBF-based safe multi-agent control}

\begin{algorithm}
  \caption{Conformal CBF-based safe controller}
	\label{code:conformal_controller}
	\begin{algorithmic}[1]
		\State {\bf input:} $x_i$, $t \in I_k$, $\{\hat{\xi}_j\}_{j \in [N_k]}$,
    $u_{ref}$, 
    $\lambda$ 
		            \State $\{\hat{x}_j\}_{j \in [N_k]} \gets $ values of $\{\hat{\xi}_j\}_{j \in [N_k]}$ at $t$
    \State $\{\hat{\dot{x}}_j\}_{j \in [N_k]} \gets $ time derivatives of $\{\hat{\xi}_j\}_{j \in [N_k]}$ at $t$
    \State $\mathcal{C}\gets$ constraints set $\left\{\hat{C}_j(x_i,\hat{x}_j,u_i,\hat{\dot{x}}_j,\lambda)\right\}_{j \in [N_k]}$
    \State $u_i\gets \underset{u_i\in U}{\mathrm{argmin}}\ \Vert u_i-u_{ref}(x_i,t)\Vert^2$,\ under the constraint $\bigwedge\mathcal{C}$\label{ln:QP_conformal_controller}
             \State {\bf return } $u_i$
	\end{algorithmic}
\end{algorithm}

Algorithm \ref{code:conformal_controller} takes as input: the ego agent's current state $x_i$, the current time $t \in I_k$, the $N_k$ sensed agents' predicted trajectories $\{\hat{\xi}_j\}_{j \in [N_k]}$ in the interval $[k\tau,(k+1)\tau)$,
the ego agent's reference controller $u_{ref}$, and the conformal control variable $\lambda$.
The algorithm is assumed to be running and generating the control value that the agent should send to its actuators instantaneously.
It proceeds as follows: it computes the time derivatives of the predicted trajectories at the current time instant $t$ and then constructs the constraints $\{\hat{C}_j\}_{j \in [N_k]}$ as in (\ref{eq:conformal_constraint}).
Finally, it solves the QP optimization problem with these  constraints.
It returns a control value that is minimally distant, in Euclidean sense, from the reference control while satisfying the constraints. % \sacha{
% Algorithm \ref{code:conformal_controller} defines a family of controllers where $\lambda$ quantifies conservativeness.
% , since the constraints $\hat C_j(\lambda)$ become tighter for lower values.}
\par

\subsection{Feasibility}

Note that even if the QP problem under the constraints $\{C_j\}_{j \in [N_k]}$ was feasible, it may not be
under the constraints $\{\hat{C}_j\}_{j \in [N_k]}$. 
We will thus assume the following.
\begin{assumption}(Feasibility of the QP problem with an  additive term)
    \label{asm:feasibility}
The QP problem in line~\ref{ln:QP_conformal_controller} of Algorithm~\ref{code:conformal_controller} has a feasible solution.
                                  \end{assumption} 

One can relax this assumption by running Algorithm~\ref{code:conformal_controller} at $t \in I_k$ with larger values of $\lambda$ than the value $\lambda_k$ determined by our update function in parallel if it is not feasible and choose the one closest to $\lambda_k$ that makes it feasible. 
However, that might affect the guarantees of the conformal controller we present next.

\subsection{Conformal CBF guarantees}
As in \cite{lekeufack2024conformal}, our conformal controller is guaranteed to have its average loss upper bounded by  $\epsilon+o(1)$.
It means that, in the long term, $s^{-1}(\epsilon)$ is the average maximum gap
% \footnote{All class-$\mathcal{K}$ functions are injective by definition.} 
between the CBF constraints we use, $\hat{C}_j$, and the ones we want to satisfy, $C_j$. Before showing that, we prove that one can choose a $\lambda$ that makes $\hat{C}_j$ strict enough to drive the loss below a pre-determined threshold immediately.   

    \begin{theorem}(Rapidly safe for any target average loss)
    Under Assumptions \ref{asm:derivative_bound}, \ref{asm:predictor}, and \ref{asm:feasibility}, our conformal controller
satisfies the following: for any $\epsilon_\mathit{safe}\in s^{-1}(\reals)$, there exists a $\lambda_\mathit{safe}\in\reals$ such that 
    for all $\lambda\leq\lambda_\mathit{safe}$, 
    %$(\{\hat{\xi}_j\}_{j \in [N]},\{\xi_j\}_{j \in [N]})\in\mathcal{X}\times\mathcal{Y}$,
    $(\{\hat{\xi}_j\}_{j \in [N]}, (\xi_i, \{\xi_j\}_{j \in [N]})) \in \mathcal{X}_p \times\mathcal{Y}$, 
   $$L(\{\hat{\xi}_j\}_{j \in [N]}, (\xi_i, \{\xi_j\}_{j \in [N]}),\lambda)\leq\epsilon_\mathit{safe}.$$
   \end{theorem}
\begin{proof}
Given $\epsilon_\mathit{safe}\in s^{-1}(\reals)$, $(\{\hat{\xi}_j\}_{j \in [N]}, (\xi_i, \{\xi_j\}_{j \in [N]}))$ and $\lambda$, pick an agent $j\in [N]$ and an instant $t \in [0,\tau)$ which maximizes $s\left(\hat{\dot{h}}_j+\alpha(\hat{h}_j) +\lambda -\dot{h}_j - \alpha(h_j)\right)$.
 Then, $\quad L(\{\hat{\xi}_j\}_{j \in [N]}, (\xi_i, \{\xi_j\}_{j \in [N]}), \lambda) \leq\epsilon_\mathit{safe}$
\begin{align*}
    &\iff \hat{\dot{h}}_j+\alpha(\hat{h}_j) +\lambda -\dot{h}_j - \alpha(h_j) \leq s^{-1}(\epsilon_\mathit{safe})\\
        &\impliedby\lambda\leq s^{-1}(\epsilon_\mathit{safe})-\lvert\hat{\dot{h}}_j-\dot{h}_j\rvert - \lvert\alpha(\hat{h}_j)-\alpha(h_j)\rvert\\
        &\impliedby\lambda\leq s^{-1}(\epsilon_\mathit{safe})-M_hE_d- M_\alpha M_h\Vert\hat{x}_j- x_j\Vert\\
    &\impliedby\lambda\leq s^{-1}(\epsilon_\mathit{safe})- M_h(E_d + M_\alpha E_v).
\end{align*}
Hence, $\lambda_\mathit{safe}\leq s^{-1}(\epsilon_\mathit{safe})- M_h(E_d + M_\alpha E_v)$ is sufficient.
\end{proof}
We now present a similar guarantee to that of Theorem \ref{thm:conformal_bound}  that our conformal controller provides.
% that is similar to that of Theorem \ref{thm:conformal_bound}.
\begin{theorem}(Long-term risk bound) 
 Fix any $\epsilon_{\mathit{safe}}$ and consider Assumptions \ref{asm:derivative_bound}, \ref{asm:predictor}, and \ref{asm:feasibility}. Then, our conformal controller  satisfies the following: if $\lambda_1\geq\lambda_\mathit{safe}-\eta$, where $\lambda_1$ is the initial value of $\lambda$ chosen by our conformal controller, and $\epsilon_{safe}\leq\epsilon$,
  then for any realization of the data, i.e., for any sequence of true and predicted trajectories, the average loss satisfies:
    $$\frac{1}{K'}\sum_{k=1}^{K'}l_k\leq\epsilon+\frac{\lambda_1-\lambda_\mathit{safe}+\eta}{\eta\cdot K'},$$
    for all $K'\in\naturals^{>0}$, where $l_k$ is the loss at the $k^{\mathit{th}}$ sampling time.
\end{theorem}
\begin{proof}
Unrolling the update rule results in the following: for all $K'\in\naturals^{>0}$, we have: $\lambda_{K'+1}=\lambda_1 +\eta\sum_{n=1}^{K'}(\epsilon-l_n).$ By isolating $\sum_{n=1}^{K'}l_n$ on one side, moving all other terms to the right-hand side, and dividing by $\eta K'$, we obtain: $\frac{1}{K'}\sum_{n=1}^{K'}l_n=\epsilon+\frac{\lambda_1-\lambda_{K'+1}}{\eta  K'}.$ To conclude, we just need to show that  $\lambda_\mathit{safe}-\eta\leq\lambda_{K'+1}$. First note that the maximal change in $\lambda$ is $\sup_k|\lambda_{k+1}-\lambda_k|<\eta$, because $l_k\in(-1/2,1/2)$ and $\epsilon\in(-1/2,1/2)$.
We will then proceed by contradiction: assume there exists a $k$ such that $\lambda_k<\lambda_\mathit{safe}-\eta$ and denote the first step that goes below the bound by $k_*$, which is equal to  $\arg\min_{k\in\naturals^{>0}}\{k\mid\lambda_k<\lambda_\mathit{safe}-\eta\}$.
Then, $k_*>1$ and, by definition of $k_*$, $\forall k<k_*,\ \lambda_{k_*}<\lambda_\mathit{safe} -\eta\leq\lambda_k$.
Because the maximum difference between successive steps is $\eta$, we have $\lambda_{k_*-1}\leq\lambda_\mathit{safe}$.
By using the update rule $\lambda_{k_*}=\lambda_{k_*-1} +\eta(\epsilon-l_{k_*-1})$, we have:
\begin{align*}
    \lambda_{k_*-1}\leq\lambda_\mathit{safe}&\implies l_{k_*-1}\leq\epsilon_\mathit{safe} 
                \implies\epsilon-l_{k_*-1}\geq\epsilon-\epsilon_\mathit{safe}\\
    &\implies\lambda_{k_*}\geq\lambda_{k_*-1}.
\end{align*}
This contradicts the minimality of $k_*$.
\end{proof}

\section{Experimental Results}
\label{sec:result}
In this section, we describe the experimental framework used to evaluate the performance of our proposed conformal controller \hussein{and analyze its results.}

We implemented Algorithm 1 and the conformal update in Python \hussein{based on the code used for} the Stanford Drone Dataset experiment in 
\cite{lekeufack2024conformal}. SDD is a large dataset of \hussein{annotated 30 fps birds-eye-view videos}  
of various types of agents (pedestrian, bicyclist, skateboarder, car, bus, and golf cart) that navigate in different high traffic areas of the campus of Stanford University~\cite{SDD_2016}. 
\hussein{We used the trajectory} predictor \hussein{\texttt{darts}~\cite{darts_paper}, the same one used in \cite{lekeufack2024conformal}, which}   
for every scenario, frame, and agent on that frame, \hussein{predicts the positions of the agent for the next 40 frames.} 
As in \cite{lekeufack2024conformal}, \hussein{a virtual robot} 
is added to the scenarios, 
equipped with the \texttt{darts} predictor, and has to reach a specific goal \hussein{position} while avoiding the real positions of \hussein{surrounding dynamic} agents. In our \hussein{experiments,} 
we only consider pedestrians, and use a \hussein{CBF-based} 
controller instead of the \hussein{model predictive one used in \cite{lekeufack2024conformal} to achieve safety.}  
We \hussein{assume that our robot has}
double integrator dynamics and \hussein{use} the \hussein{artificial} potential fields-based CBF from \cite{Artificial_Potential_Fields} \hussein{for collision avoidance}. The state of the robot is a tuple $\mathbf{X}_i=(x,y,v_x,v_y)^\intercal$, where $X_i=(x,y)$ represents its position \hussein{and $v_i=(v_x,v_y)$ is its linear velocity}, and \hussein{its control input is the linear acceleration:}
$\mathbf{u}=(u_x,u_y)^\intercal$. \hussein{Its dynamics are defined as follows:}
    $\dot{\mathbf{X}}_i= (v_x,v_y,u_x,u_y)^\intercal.$
    \hussein{We consider the} reference control \hussein{to be the gradient of a virtual}  
attractive field centered \hussein{at} 
the goal: $\nabla U_{att}:=K_{att}*(X_i-X_{goal})$. \hussein{Following \cite{Artificial_Potential_Fields}, we define the CBF for a given surrounding agent $j$ as follows:} $h_j(\mathbf{X}_i)=\frac{1}{1+U_{rep,j}(X_i)}-\delta$, 
\hussein{where} $U_{rep,j}=\frac{K_{rep}}{2}\left(\frac{1}{\Vert X_i-X_j\Vert}-\frac{1}{\rho_0}\right)^2$ \hussein{is a virtual repulsive field centered at agent $j$'s position, and} 
$\delta$ and $\rho_0$ represent  the maximal safe repulsion (i.e., at the collision distance) and the sensing distance of the ego agent, respectively. Since the CBF only uses the distance $\Vert X_i-X_j\Vert$ and not the ego agent's speed, we \hussein{define} 
the QP problem \hussein{with the} 
desired safe speed $v_{QP}$ \hussein{as its optimization variable} and set $\mathbf{u}_{QP}=-K_{acc}(v_i - v_{QP})$ \hussein{as a  tracking controller of  $v_{QP}$, as in \cite{Artificial_Potential_Fields}. It has a separate safety constraint for every surrounding agent}.\par
\begin{table}[h]
\small
    \centering
    \begin{tabular}{|*{5}{c|}}
    \hline
    \multicolumn{1}{|c|}{Parameters} & \multicolumn{4}{|c|}{Metrics}\\ \hline
    $\epsilon$ & $t_{goal}$ & $n_{collide}$ & $d_{\min}$ & $l_{avg}$ \\ \hline
    -0.4 & 19.2 & 0 & 4.098 & -0.3949\\ \hline
	-0.2 & 19.2 & 0 & 4.098 & -0.1865\\ \hline
	0 & 19.2 & 4 & 0.9031 & 0.008016\\ \hline
	0.2 & 18.63 & 0 & 1.639 & 0.2008\\ \hline
	0.4 & 14.63 & 18 & 0.4813 & 0.4007\\ \hline
	pred/noLearn & 13.9 & 14 & 0.8311 & -\\ \hline
	gt/noLearn & 13.9 & 13 & 0.7764 & - \\ \hline
\end{tabular}
    \caption{gates\_1 scenario with $a=0.1$, $K_{acc}=2$, $K_{rep}=20$, $K_{att}=1$, $\rho_0=400$, $\tau=12$, and $\eta=100$.}
    \label{tab:eps}
\end{table}
There are several parameters we experimented with to evaluate the performance of our conformal controller. They can be divided into two groups: the controller's inherent parameters (the sensor's sampling time $\tau$, the extended class-$\mathcal{K}$ function $\alpha$, which we assume to be linear with slope $a$, as well as the potential fields parameters' $K_{att}$ and $K_{rep}$), and the learning parameters (the learning rate $\eta$, and the target loss $\epsilon$).
In order to assess the controller's performance, we used 
several key metrics: the time for reaching the goal $t_{goal}$,  the number of  frames at which the agent was at collision $n_{collide}$, the minimal distance between the ego agent and the pedestrians $d_{\min}$, and finally the average loss of our conformal controller $l_{avg}$.
Recall that our approach does not guarantee \hussein{collision avoidance,} 
but a bound on the \hussein{average} loss
$\hat{\dot{h}}_j+\alpha(\hat{h}_j)+\lambda-\dot{h}_j-\alpha(h_j)$ \hussein{that is a surrogate for the average number of CBF constraints violations.}
\hussein{We evaluate our controller against} two baseline controllers, "pred/noLearn" and "gt/noLearn". Both run Algorithm 1 without learning, i.e., with $\lambda$ \hussein{and $\eta$} equal to zero, \hussein{but one has access to} 
the same black-box predictor and \hussein{the other has access to the}
ground truth (i.e., a perfect predictor).\par
We first evaluate the influence of the target average loss $\epsilon$. Table I presents the performance of the controller for different $\epsilon$, fixing all other parameters. We see that \hussein{adapting $\lambda$ based on prediction errors}  
reduces the number of safety violations and increases the minimum distance overall. The average loss \hussein{in all experiments is very close to $\epsilon$, as intended.} \hussein{Moreover,} compared to the baselines, decreasing $\epsilon$ leads to \hussein{fewer collisions. We also observe that even with access to the ground truth future states of surrounding agents, "gt/noLearn" still led to 12 collision frames. That is because of tracking errors, as the QP solver outputs target safe velocities, which are then tracked by a proportional controller that generates accelerations, instead of generating the control input directly. Having negative $\epsilon$, which results in the CBF constraint being more restrictive than the ground truth one on average decreases $n_{\mathit{collide}}$ to zero. } 
\begin{table}[!h]
\small 
    \centering
    \begin{tabular}{|*{6}{c|}}
    \hline
    \multicolumn{2}{|c|}{Parameters} & \multicolumn{4}{|c|}{Metrics}\\ \hline
    $\tau$ & $\eta$ & $t_{goal}$ & $n_{collide}$ & $d_{\min}$ & $l_{avg}$ \\ \hline
    4 & 0.1 & 9.067 & 6 & 0.4059 & -0.2052\\ \hline
	4 & 1 & 9.1 & 2 & 0.8881 & -0.2415\\ \hline
	4 & noLearn & 8.933 & 13 & 0.4493 & -\\ \hline
	8 & 0.1 & 9.067 & 6 & 1.046 & -0.1949\\ \hline
	8 & 1 & 9.133 & 3 & 0.4257 & -0.2284\\ \hline
	8 & noLearn & 8.933 & 13 & 0.3877 & -\\ \hline
	gt & noLearn & 8.967 & 11 & 0.5087 & -\\ \hline
\end{tabular}
    \caption{nexus\_4 scenario with $a=1$, $K_{acc}=2$, $K_{rep}=40$, $K_{att}=1$, $\rho_0=400$, and $\epsilon=-0.25$.}
    \label{tab:tau_eta}
\end{table}
For the second experiment, we compare different values of $\tau$ and $\eta$ for a fixed \hussein{$\epsilon$.} 
The results in table \ref{tab:tau_eta} indicate that a higher learning rate and a lower sampling time result in fewer collisions and larger minimum distances. A large $\eta$ and a small $\tau$ mean that the agent is able to change $\lambda$ faster, leading to a lower loss. Note that Assumption~\ref{asm:feasibility} was only satisfied in this case  when $\eta$ was smaller than or equal to one, because the QP problem became infeasible at some time instances with larger $\eta$.

\section{Conclusion}
\label{sec:conclusion}
We presented an algorithm to adapt the CBF-based inter-agent collision-avoidance safety constraints in decentralized multi-agent settings according to prediction errors. Our adaptation aims to increase restrictiveness of the constraints when they turn out to be looser than intended by the user, and vice versa. Using conformal decision theory, we obtained an upper bound on the average value of a monotonic function of the difference between our constraints and the ground truth ones. Finally, we presented experimental results validating our theoretical results and comparing the effects of different hyperparameters.  
\clearpage

\bibliographystyle{IEEEtran}
\bibliography{references}

\begin{thebibliography}{10}
\providecommand{\url}[1]{#1}
\csname url@samestyle\endcsname
\providecommand{\newblock}{\relax}
\providecommand{\bibinfo}[2]{#2}
\providecommand{\BIBentrySTDinterwordspacing}{\spaceskip=0pt\relax}
\providecommand{\BIBentryALTinterwordstretchfactor}{4}
\providecommand{\BIBentryALTinterwordspacing}{\spaceskip=\fontdimen2\font plus
\BIBentryALTinterwordstretchfactor\fontdimen3\font minus \fontdimen4\font\relax}
\providecommand{\BIBforeignlanguage}[2]{{%
\expandafter\ifx\csname l@#1\endcsname\relax
\typeout{** WARNING: IEEEtran.bst: No hyphenation pattern has been}%
\typeout{** loaded for the language `#1'. Using the pattern for}%
\typeout{** the default language instead.}%
\else
\language=\csname l@#1\endcsname
\fi
#2}}
\providecommand{\BIBdecl}{\relax}
\BIBdecl

\bibitem{10.1145/3582576}
\BIBentryALTinterwordspacing
E.~Vinitsky, N.~Lichtl\'{e}, K.~Parvate, and A.~Bayen, ``Optimizing mixed autonomy traffic flow with decentralized autonomous vehicles and multi-agent reinforcement learning,'' \emph{ACM Trans. Cyber-Phys. Syst.}, vol.~7, no.~2, apr 2023. [Online]. Available: \url{https://doi.org/10.1145/3582576}
\BIBentrySTDinterwordspacing

\bibitem{SAMADI2020106211}
\BIBentryALTinterwordspacing
E.~Samadi, A.~Badri, and R.~Ebrahimpour, ``Decentralized multi-agent based energy management of microgrid using reinforcement learning,'' \emph{International Journal of Electrical Power \& Energy Systems}, vol. 122, p. 106211, 2020. [Online]. Available: \url{https://www.sciencedirect.com/science/article/pii/S0142061520304877}
\BIBentrySTDinterwordspacing

\bibitem{9926465}
J.~Bidmead, S.~Bhatiani, and X.~Xu, ``Decentralized factory control based on multi-agent technologies,'' in \emph{2022 IEEE 18th International Conference on Automation Science and Engineering (CASE)}, 2022, pp. 307--311.

\bibitem{8643440}
M.~K. Al-Sharman, Y.~Zweiri, M.~A.~K. Jaradat, R.~Al-Husari, D.~Gan, and L.~D. Seneviratne, ``Deep-learning-based neural network training for state estimation enhancement: Application to attitude estimation,'' \emph{IEEE Transactions on Instrumentation and Measurement}, vol.~69, no.~1, pp. 24--34, 2020.

\bibitem{GAN}
A.~Gupta, J.~Johnson, L.~Fei-Fei, S.~Savarese, and A.~Alahi, ``Social gan: Socially acceptable trajectories with generative adversarial networks,'' in \emph{2018 IEEE/CVF Conference on Computer Vision and Pattern Recognition}, 2018, pp. 2255--2264.

\bibitem{li2020gripplus}
X.~Li, X.~Ying, and M.~C. Chuah, ``Grip++: Enhanced graph-based interaction-aware trajectory prediction for autonomous driving,'' 05 2020.

\bibitem{pedestrian_trajectory_prediction_using_CNNs_2022}
\BIBentryALTinterwordspacing
S.~Zamboni, Z.~T. Kefato, S.~Girdzijauskas, C.~Norén, and L.~{Dal Col}, ``Pedestrian trajectory prediction with convolutional neural networks,'' \emph{Pattern Recognition}, vol. 121, p. 108252, 2022. [Online]. Available: \url{https://www.sciencedirect.com/science/article/pii/S0031320321004325}
\BIBentrySTDinterwordspacing

\bibitem{CBF_based_quadratic_programs_2017_AaronAmes}
A.~D. Ames, X.~Xu, J.~W. Grizzle, and P.~Tabuada, ``Control barrier function based quadratic programs for safety critical systems,'' \emph{IEEE Transactions on Automatic Control}, vol.~62, no.~8, pp. 3861--3876, 2017.

\bibitem{lekeufack2024conformal}
J.~Lekeufack, A.~N. Angelopoulos, A.~Bajcsy, M.~I. Jordan, and J.~Malik, ``Conformal decision theory: Safe autonomous decisions from imperfect predictions,'' in \emph{2024 IEEE International Conference on Robotics and Automation (ICRA)}.\hskip 1em plus 0.5em minus 0.4em\relax IEEE, 2024, pp. 11\,668--11\,675.

\bibitem{conformal_prediction_original_gammerman_1998}
A.~Gammerman, V.~Vovk, and V.~Vapnik, ``Learning by transduction,'' in \emph{Proceedings of the Fourteenth Conference on Uncertainty in Artificial Intelligence}, ser. UAI'98.\hskip 1em plus 0.5em minus 0.4em\relax San Francisco, CA, USA: Morgan Kaufmann Publishers Inc., 1998, p. 148–155.

\bibitem{conformal_prediction_intro_anastasios_2023}
\BIBentryALTinterwordspacing
A.~N. Angelopoulos and S.~Bates, ``Conformal prediction: A gentle introduction,'' \emph{Found. Trends Mach. Learn.}, vol.~16, no.~4, p. 494–591, mar 2023. [Online]. Available: \url{https://doi.org/10.1561/2200000101}
\BIBentrySTDinterwordspacing

\bibitem{conformal_decision_making_vovk_2018}
V.~Vovk and C.~Bendtsen, ``Conformal predictive decision making,'' in \emph{Conformal and Probabilistic Prediction and Applications}.\hskip 1em plus 0.5em minus 0.4em\relax PMLR, 2018, pp. 52--62.

\bibitem{adaptive_conformal_inference_under_distribution_shift_gibbs_neurips_2021}
\BIBentryALTinterwordspacing
I.~Gibbs and E.~Candes, ``Adaptive conformal inference under distribution shift,'' in \emph{Advances in Neural Information Processing Systems}, A.~Beygelzimer, Y.~Dauphin, P.~Liang, and J.~W. Vaughan, Eds., 2021. [Online]. Available: \url{https://openreview.net/forum?id=6vaActvpcp3}
\BIBentrySTDinterwordspacing

\bibitem{dixit2022adaptive}
A.~Dixit, L.~Lindemann, S.~Wei, M.~Cleaveland, G.~J. Pappas, and J.~W. Burdick, ``Adaptive conformal prediction for motion planning among dynamic agents,'' 2022.

\bibitem{safety_assurance_using_conformal_prediction_2023}
R.~Luo, S.~Zhao, J.~Kuck, B.~Ivanovic, S.~Savarese, E.~Schmerling, and M.~Pavone, ``Sample-efficient safety assurances using conformal prediction,'' in \emph{Algorithmic Foundations of Robotics XV}, S.~M. LaValle, J.~M. O'Kane, M.~Otte, D.~Sadigh, and P.~Tokekar, Eds.\hskip 1em plus 0.5em minus 0.4em\relax Cham: Springer International Publishing, 2023, pp. 149--169.

\bibitem{SDD_2016}
\BIBentryALTinterwordspacing
A.~Robicquet, A.~Sadeghian, A.~Alahi, and S.~Savarese, ``Learning social etiquette: Human trajectory understanding in crowded scenes,'' in \emph{European Conference on Computer Vision}, 2016. [Online]. Available: \url{https://api.semanticscholar.org/CorpusID:3150075}
\BIBentrySTDinterwordspacing

\bibitem{Risk-Aware-CBF}
M.~Black, G.~Fainekos, B.~Hoxha, D.~Prokhorov, and D.~Panagou, ``Safety under uncertainty: Tight bounds with risk-aware control barrier functions,'' in \emph{2023 IEEE International Conference on Robotics and Automation (ICRA)}, 2023, pp. 12\,686--12\,692.

\bibitem{onlineCBFdecenter}
Z.~Gao, G.~Yang, and A.~Prorok, ``Online control barrier functions for decentralized multi-agent navigation,'' in \emph{2023 International Symposium on Multi-Robot and Multi-Agent Systems (MRS)}, 2023, pp. 107--113.

\bibitem{rate_tune_CBF}
H.~Parwana, A.~Mustafa, and D.~Panagou, ``Trust-based rate-tunable control barrier functions for non-cooperative multi-agent systems,'' in \emph{2022 IEEE 61st Conference on Decision and Control (CDC)}, 2022, pp. 2222--2229.

\bibitem{huang2023conformal}
H.~Huang, S.~Sharma, A.~Loquercio, A.~Angelopoulos, K.~Goldberg, and J.~Malik, ``Conformal policy learning for sensorimotor control under distribution shifts,'' 2023.

\bibitem{khalil2002nonlinear}
\BIBentryALTinterwordspacing
H.~Khalil, \emph{Nonlinear Systems}, ser. Pearson Education.\hskip 1em plus 0.5em minus 0.4em\relax Prentice Hall, 2002. [Online]. Available: \url{https://books.google.com/books?id=t_d1QgAACAAJ}
\BIBentrySTDinterwordspacing

\bibitem{nagumo_thrm}
F.~Blanchini and S.~Miani, \emph{Set-Theoretic Methods in Control}, 01 2007.

\bibitem{robust_CBF_IFAC_2015}
\BIBentryALTinterwordspacing
X.~Xu, P.~Tabuada, J.~W. Grizzle, and A.~D. Ames, ``Robustness of control barrier functions for safety critical control,'' \emph{IFAC-PapersOnLine}, vol.~48, no.~27, pp. 54--61, 2015, analysis and Design of Hybrid Systems ADHS. [Online]. Available: \url{https://www.sciencedirect.com/science/article/pii/S2405896315024106}
\BIBentrySTDinterwordspacing

\bibitem{Comparison_CBF_APF_2020_AaronAmes}
A.~W. Singletary, K.~Klingebiel, J.~R. Bourne, N.~A. Browning, P.~T. Tokumaru, and A.~D. Ames, ``Comparative analysis of control barrier functions and artificial potential fields for obstacle avoidance,'' \emph{2021 IEEE/RSJ International Conference on Intelligent Robots and Systems (IROS)}, pp. 8129--8136, 2020.

\bibitem{conf_pred_time_series}
C.~Xu and Y.~Xie, ``Conformal prediction for time series,'' \emph{IEEE Transactions on Pattern Analysis and Machine Intelligence}, vol.~45, no.~10, pp. 11\,575--11\,587, oct 2023.

\bibitem{CBF_theory}
A.~D. Ames, S.~Coogan, M.~Egerstedt, G.~Notomista, K.~Sreenath, and P.~Tabuada, ``Control barrier functions: Theory and applications,'' in \emph{2019 18th European Control Conference (ECC)}, 2019, pp. 3420--3431.

\bibitem{deeplearning_predictor_multiagent}
\BIBentryALTinterwordspacing
X.~Wang, H.~Pu, H.~J. Kim, and H.~Li, ``Deepsafempc: Deep learning-based model predictive control for safe multi-agent reinforcement learning,'' 2024. [Online]. Available: \url{https://arxiv.org/abs/2403.06397}
\BIBentrySTDinterwordspacing

\bibitem{multi_rl_network}
T.~Chu, S.~Chinchali, and S.~Katti, ``Multi-agent reinforcement learning for networked system control,'' \emph{arXiv preprint arXiv:2004.01339}, 2020.

\bibitem{Trajectory_Multiagent_Distributed}
L.~Gonçalves and A.~Schoellig, ``Trajectory generation for multiagent point-to-point transitions via distributed model predictive control,'' \emph{IEEE Robotics and Automation Letters}, vol.~PP, pp. 1--1, 04 2019.

\bibitem{Amirloo2022LatentFormerMT}
\BIBentryALTinterwordspacing
E.~Amirloo, A.~Rasouli, P.~Lakner, M.~Rohani, and J.~Luo, ``Latentformer: Multi-agent transformer-based interaction modeling and trajectory prediction,'' \emph{ArXiv}, vol. abs/2203.01880, 2022. [Online]. Available: \url{https://api.semanticscholar.org/CorpusID:247223143}
\BIBentrySTDinterwordspacing

\bibitem{Linear_predictors}
\BIBentryALTinterwordspacing
M.~Korda and I.~Mezić, ``Linear predictors for nonlinear dynamical systems: Koopman operator meets model predictive control,'' \emph{Automatica}, vol.~93, pp. 149--160, 2018. [Online]. Available: \url{https://www.sciencedirect.com/science/article/pii/S000510981830133X}
\BIBentrySTDinterwordspacing

\bibitem{darts_paper}
\BIBentryALTinterwordspacing
J.~Herzen, F.~Lässig, S.~G. Piazzetta, T.~Neuer, L.~Tafti, G.~Raille, T.~V. Pottelbergh, M.~Pasieka, A.~Skrodzki, N.~Huguenin, M.~Dumonal, J.~Kocisz, D.~Bader, F.~Gusset, M.~Benheddi, C.~Williamson, M.~Kosinski, M.~Petrik, and G.~Grosch, ``Darts: User-friendly modern machine learning for time series,'' \emph{Journal of Machine Learning Research}, vol.~23, no. 124, pp. 1--6, 2022. [Online]. Available: \url{http://jmlr.org/papers/v23/21-1177.html}
\BIBentrySTDinterwordspacing

\bibitem{Artificial_Potential_Fields}
A.~Singletary, K.~Klingebiel, J.~Bourne, A.~Browning, P.~Tokumaru, and A.~Ames, ``Comparative analysis of control barrier functions and artificial potential fields for obstacle avoidance,'' in \emph{2021 IEEE/RSJ International Conference on Intelligent Robots and Systems (IROS)}, 2021, pp. 8129--8136.

\end{thebibliography}
\end{document}